\documentclass{amsart}
\usepackage{amsfonts}
\usepackage{amsmath,amscd,lscape}
\usepackage{amsthm}
\usepackage{amssymb}
\usepackage{latexsym}
\usepackage{mathrsfs}
\usepackage{tikz}

\usepackage[utf8]{inputenc}

\usepackage{graphicx}
\usepackage{subcaption}
\usepackage{float}

\newcommand{\nc}{\newcommand}
\nc{\ben}{\begin{eqnarray}}
\nc{\een}{\end{eqnarray}}

\newcommand{\beqa}{\begin{eqnarray}}
\newcommand{\eeqa}{\end{eqnarray}}
\nc{\Z}{{\bold Z}}

\newcommand{\fpt}[7]{{}_4\phi_3\left[\begin{matrix} #1 , #2, #3, #4 \\
#5, #6, #7 \end{matrix}\,; q^2,q^2\right]}

\usepackage[curve,matrix,arrow,color]{xy}

\usepackage{color}

\newtheorem{lem}{Lemma}[section]

\newtheorem{prop}{Proposition}[section]

\newtheorem{defn}{Definition}[section]

\newtheorem{hyp}{Hypothesis}

\newcommand{\cal}{\mathcal}

\newcommand{\tA}{\textsf{A}}

\newcommand{\tW}{\textsf{W}}

\newcommand{\tb}{\mathsf{b}} 
\newcommand{\tc}{\mathsf{c}}

\newcommand{\cV}{\mathcal{V}}

\newcommand{\cW}{\mathcal{W}}

\newcommand{\non}{\nonumber}

\numberwithin{equation}{section}
\begin{document}

\title[The $q$-Racah polynomials from scalar products of Bethe states
%A determinant representation for $q$-Racah polynomials
%Leonard pairs's transition matrix and %$q$-Racah polynomials 
%Bethe states
]{The $q$-Racah polynomials from scalar products of Bethe states
%Scalar products of Bethe states and \\ Leonard pairs transition matrix
%$q$-Racah orthogonal  polynomials \\
%A determinant representation for the $q$-Racah polynomials\\ from the algebraic Bethe ansatz
}
%\dedicatory{}
\author{Pascal Baseilhac$^{*}$}
\address{$^*$ Institut Denis-Poisson CNRS/UMR 7013 - Universit\'e de Tours - Universit\'e d'Orl\'eans
Parc de Grammont, 37200 Tours, 
FRANCE}
\email{pascal.baseilhac@idpoisson.fr}

\author{Rodrigo A. Pimenta$^{**}$}
\address{$^{**}$ Department of Physics and Astronomy, University of Manitoba, Winnipeg, R3T 2N2, CANADA} 
\email{rodrigo.alvespimenta@umanitoba.ca}

\begin{abstract} The $q$-Racah polynomials are expressed in terms of certain ratios of scalar products of Bethe states associated with
Bethe equations of either homogeneous or inhomogeneous type. This result is obtained by combining the theory of
Leonard pairs and the modified algebraic Bethe ansatz.  
\end{abstract}

\maketitle

	\vskip -0.5cm
	


\section{Introduction}
At the top of the hierarchy of  classical $q$-hypergeometric orthogonal polynomials known as the discrete Askey scheme \cite{AW79,KS96}, the $q$-Racah polynomials satisfy a bispectral property that is encoded into the representation theory of 
the Askey–Wilson algebra \cite{T87,Z91} using the theory of Leonard pairs \cite{T04,Ter04}. More generally, it is known that 
the entries of the transition matrices relating different eigenbases of elements of a Leonard pair correspond to discrete orthogonal polynomials \cite{T03}.  All other families of orthogonal polynomials of the scheme can be reached from the $q$-Racah polynomials by various limit transitions (either in the scalar parameters entering into the definition of the polynomials, or $q\rightarrow 1$). See \cite{KS96,Ko10} for details.
\vspace{1mm}

Recently, the spectral problem for certain combinations of elements of a Leonard pair of $q$-Racah type - known as the Heun-Askey-Wilson operator - has been solved using the theory of Leonard pairs combined with the framework of the modified algebraic Bethe ansatz \cite{BP19}.
The modified algebraic Bethe ansatz
was developed to analyze the eigenproblem of quantum spin chains with U(1) symmetry breaking boundary fields \cite{BC13,B15,C15,ABGP15,BP15}. In this approach, the eigenvalues are expressed in terms of Bethe roots satisfying a system of transcendental equations known as the Bethe equations whereas the eigenvectors are Bethe states. In some instances, the Bethe equations are of the inhomogeneous type \cite{ODBAbook}. For the case $q=1$, see the analog analysis in the recent works \cite{Nico1,Nico2}. As observed in \cite{BP19}, some examples of eigenbases for Leonard pairs of $q$-Racah type can be constructed as Bethe states.\vspace{1mm}

The purpose of this letter is to study in details different types of eigenbases and dual eigenbases of Leonard pairs of $q$-Racah type using the framework of the modified algebraic Bethe ansatz.  Remarkably, the eigenbases can be built from Bethe states associated with homogeneous or inhomogeneous Bethe equations.
This leads to an interpretation of the $q$-Racah polynomials as certain ratios of scalar products of
various types of Bethe states. \vspace{1mm}  

The text is organized as follows. In Section 2, the concept of Leonard pairs, related eigenbases and the interpretation
of the correspondence between the entries of the transition matrix and the $q$-Racah polynomials are recalled. In Section 3, 
 eigenbases and dual eigenbases for a Leonard pair of $q$-Racah type  are constructed in terms of Bethe states
 and dual Bethe states associated with Bethe equations of either homogeneous or inhomogeneous type. 
In particular, this leads to certain identities relating Bethe states of homogeneous type and Bethe states of inhomogenous type. 
Using the theory of Leonard pairs, a correspondence between  $q$-Racah orthogonal polynomials and ratios of scalar products of Bethe states follows, see Subsection \ref{sec:qrac}.
Concluding remarks are given in the last section. Part of the necessary material for the proofs can be found in \cite{BP19} and Appendix A.

\vspace{3mm}

{\bf Notations:} The parameter $q$ is assumed not to be a root of unity and $q \neq 1$. We write $[X,Y]_q = qXY - q^{-1}YX$ and $[n]_q=\frac{q^n-q^{-n}}{q-q^{-1}}$.
 The identity element is denoted $\mathcal{I}$. We use the standard $q$-shifted factorials:
\beqa
(a;q)_n = \prod_{k=0}^{n-1}(1-aq^k)\ ,\qquad (a_1,a_2,\cdots,a_k;q)_n = \prod_{j=1}^{k}(a_j;q)\ ,
\eeqa
and define
\beqa
b(x)=x-x^{-1}\,\,.\label{b} 
\eeqa

\section{The Askey-Wilson algebra,  Leonard pairs and transition matrix %$q$-Racah polynomials
}\label{sec2}
In this section, we introduce the defining relations of the Askey-Wilson algebra with generators  $\tA,\tA^*$.  The concept of a Leonard pair associated with $\tA,\tA^*$ is briefly reviewed. %Introducing an irreducible finite dimensional vector space for AW denoted $\cal V$ and its dual $ \tilde{\cal V}$, 
The eigenbases and dual eigenbases for the Leonard pair are defined.  Then, we recall the transition matrices relating those eigenbases, and how they are expressed in terms of the $q$-Racah orthogonal polynomials and scalar products of Leonard pairs' eigenvectors \cite{Z91,T03,T04,Ter04}.

\subsection{Askey-Wilson algebra and Leonard pairs}
Let $\rho,\omega,\eta,\eta^* \in {\mathbb C}^*$ be generic. The  Askey-Wilson algebra (AW) is generated by $\tA,\tA^*$ subject to the relations \cite{T87,Z91}
\beqa
&&\big[\tA,\big[\tA,\tA^*\big]_q\big]_{q^{-1}}=  \rho \,\tA^*+\omega \,\tA+\eta\mathcal{I} \ , \label{aw1} \\
&&\big[\tA^*,\big[\tA^*,\tA\big]_q\big]_{q^{-1}}= \rho \,\tA+\omega \,\tA^*+\eta^*\mathcal{I} \ .\label{aw2}
\eeqa
%
%For $q$ not a root of unity and finite dimensional representations, the representation theory of the Askey-Wilson algebra is described by the theory of Leonard pairs, that is now recalled.\vspace{1mm}

Let $\cV$ be a vector space of positive finite dimension $\dim({\cV})=2s+1$, where $s$ is an integer or half-integer. Define the algebra homomorphism $\pi: {\rm AW} \rightarrow {\rm End}(\cal V)$. Assume 
%$q$ is not a root of unity,  
$\pi(\tA),\pi(\tA^*)$ are diagonalizable on  ${\cV}$, each  multiplicity-free and ${\cV}$  is irreducible. Then, by  \cite[Theorem 6.2]{T04},  $\pi(\tA),\pi(\tA^*)$ is a Leonard pair \cite[Definition 1.1]{T04}. Given  the eigenvalue sequence $\{\theta_M\}_{M=0}^{2s}$ associated with  $\pi(\tA)$  (resp.  the eigenvalue sequence $\{\theta_N^*\}_{N=0}^{2s}$ associated with  $\pi(\tA^*)$), one associates an eigenbasis with vectors  $\{|\theta_M \rangle \}_{M=0}^{2s}$ (resp. an eigenbasis with vectors $\{|\theta^*_N\rangle \}_{N=0}^{2s}$). For a Leonard pair, recall that: (i)  in the eigenbasis of  $\pi(\tA)$, then  $\pi(\tA^*)$ acts as a tridiagonal matrix; (ii) in the eigenbasis of  $\pi(\tA^*)$, then  $\pi(\tA)$ acts as a tridiagonal matrix:
\beqa
\qquad \quad \pi(\tA) |\theta_M\rangle &=& \theta_M|\theta_M\rangle \ ,\quad  \pi(\tA^*) |\theta_M\rangle =
A^*_{M+1,M} |\theta_{M+1}\rangle +  A^*_{M,M} |\theta_{M}\rangle +  A^*_{M-1,M}|\theta_{M-1}\rangle \ ,\label{tridAstar}\\
\qquad \quad  \pi(\tA^*) |\theta^*_N \rangle &=&  \theta^*_N |\theta^*_N \rangle  \ ,\quad  \pi(\tA)|\theta^*_N \rangle   =
 A_{N+1,N} |\theta^*_{N+1} \rangle  +   A_{N,N} |\theta^*_N \rangle  +   A_{N-1,N}|\theta^*_{N-1} \rangle \ . \label{tridAstar2}
\eeqa
Here $A^*_{-1,0}=A^*_{2s+1,2s}=A_{-1,0}=A_{2s+1,2s}=0$, and the explicit expressions for the coefficients $\{A^*_{M\pm 1,M},A^*_{M,M}\}$ and $\{A_{N\pm 1,N},A_{N,N}\}$ in terms of the eigenvalue sequences are given in \cite{Ter04}.\vspace{1mm}

Let $\tilde\cV$ be the dual vector space of $\cV$, i.e. the vector space  $\tilde\cV$ of all linear functionals from $\cV$ to ${\mathbb C}$. Define the family of covectors $\{\langle \theta_M |\}_{M=0}^{2s}\in \tilde\cV$ (resp. $\{\langle \theta^*_M |\}_{N=0}^{2s}\in \tilde\cV$) associated with the eigenvalue sequence $\{\theta_M\}_{M=0}^{2s}$ (resp.   $\{\theta_N^*\}_{N=0}^{2s}$) such that:
\beqa
\qquad \quad \langle \theta_M |\pi(\tA)  &=& \langle \theta_M |\theta_M\ ,\quad  \langle \theta_M |\pi(\tA^*) =\langle \theta_{M+1 }|
\tilde{A}^*_{M,M+1}  +  \langle \theta_{M }|\tilde{A}^*_{M,M}  + \langle \theta_{M-1 } |\tilde{A}^*_{M,M-1} \ ,\label{tridAstardual}\\
\qquad \quad \langle \theta_N^* |\pi(\tA^*)  &=& \langle \theta_N^* |\theta_N^*\ ,\quad  \langle \theta_N^* |\pi(\tA) =\langle \theta_{N+1 }^*|
\tilde{A}_{N,N+1}  +  \langle \theta_{N }^*|\tilde{A}_{N,N}  + \langle \theta_{N-1 } |\tilde{A}_{N,N-1}\ . 
\label{tridAstar2dual}
\eeqa
The coefficients $\tilde{A}^*_{M,M'},\tilde{A}_{N,N'}$ are determined as follows.
Let us associate a column to a vector $|v\rangle \in \cV$ and a row  to a covector $\langle \tilde{v}|\in \tilde{\cV}$. Then,  one introduces the scalar product $\langle .|. \rangle : \tilde{\cV} \times \cV \rightarrow {\mathbb C}$. The representation is irreducible, so $\langle .|. \rangle$ is non-degenerate\footnote{There are no non-trivial subspaces  $\tilde\cW = \{\tilde{w} \in \tilde{\cV} | \langle \tilde{w}  | \cV \rangle = 0\}$ and $\cW = \{w \in \cV | \langle \tilde{\cV}| w \rangle = 0\}$).
% \rodrigo{if $\langle \tilde{w}|w\rangle = 0$ for all $|w\rangle \in \cV$ then $\langle \tilde{w}|=0$; if $\langle \tilde{w}|w\rangle = 0$ for all $\langle \tilde{w}| \in \tilde{\cV}$ then $|w\rangle=0$.}
}.
Using (\ref{tridAstar}), (\ref{tridAstardual}) and (\ref{tridAstar2}), (\ref{tridAstar2dual}), the equalities $(\langle \theta_{M'} |\pi(\tA))| \theta_M \rangle = \langle \theta_{M'} |(\pi(\tA)| \theta_M \rangle)$ and $(\langle \theta^*_{N'} |\pi(\tA^*))| \theta_N^* \rangle = \langle \theta_{N'}^* |(\pi(\tA^*)| \theta_N^* \rangle)$
imply the orthogonality relations
\beqa
\langle \theta_{M'} |\theta_M \rangle  = \delta_{MM'}\xi_M \ ,\qquad \langle \theta_{N'}^* |\theta_N^* \rangle  = \delta_{NN'}\xi_N^* \label{orthcond}\ ,
\eeqa
where $\delta_{ij}$ denotes the Kronecker symbol and the normalization factors  $\xi_M \neq 0$, $\xi^*_N\neq 0$ are introduced.
Now, using (\ref{tridAstar}), (\ref{tridAstardual}) and (\ref{tridAstar2}), (\ref{tridAstar2dual}) together with (\ref{orthcond}), the equalities $(\langle \theta_{M'} |\pi(\tA^*))| \theta_M \rangle = \langle \theta_{M'} |(\pi(\tA^*)| \theta_M \rangle)$ and $(\langle \theta^*_{N'} |\pi(\tA))| \theta_N^* \rangle = \langle \theta_{N'}^* |(\pi(\tA)| \theta_N^* \rangle)$ imply:
\beqa
 \tilde{A}^*_{M,M'} &=&  A^*_{M,M'} \frac{\xi_{M}}{\xi_{M'}}\quad \mbox{for $M'=M+1,M,M-1$}\ ,\\
\tilde{A}_{N,N'} &=&  A_{N,N'} \frac{\xi_{N}^*}{\xi_{N'}^*} \quad \mbox{for $N'=N+1,N,N-1$}\ .
\eeqa

In the following, we consider  the eigenvalue sequences of the form \cite[Theorem 4.4 (case I)]{Ter03}:
\beqa
 \theta_M= \tb q^{2M} + \tc q^{-2M}\ , \quad   \theta^*_N= \tb^* q^{2N} + \tc^*q^{-2N}\ ,\label{st}
\eeqa
where $\tb,\tc,\tb^*,\tc^* \in {\mathbb C}^*$.
For this parametrization, the structure constant $\rho$ in (\ref{aw1}), (\ref{aw2}) is given by \cite[Lemma 4.5]{Ter03}:
\beqa
\rho=-\tb\tc(q^2-q^{-2})^2=-\tb^*\tc^*(q^2-q^{-2})^2
% =-b^\diamond c^\diamond(q^2-q^{-2})^2
\ .\label{rho}
\eeqa
Without loss of generality, the other structure constants can be written in the form%(see e.g. \cite[eq. (3.32)-(3.36)]{BP19})
:
\beqa
\omega&=& (q-q^{-1})^2 \left( \tb\tc (\zeta^2+\zeta^{-2})(q^{2s+1} + q^{-2s-1}) - (\tb q^{2s}+\tc q^{-2s}) (\tb^*q^{2s}+\tc^*q^{-2s}) \right)\ ,\label{scpar1}\\
\eta&=& -\frac{(q^2-q^{-2})^2}{(q+q^{-1})}  \tb\tc  \left(  (\tb q^{2s}+\tc q^{-2s})(\zeta^2+\zeta^{-2}) - (\tb^*q^{2s}+\tc^*q^{-2s})(q^{2s+1} + q^{-2s-1}) \right)\ ,\label{scpar2}\\
\eta^*&=& -\frac{(q^2-q^{-2})^2}{(q+q^{-1})}  \tb^*\tc^*  \left(  (\tb^*q^{2s}+\tc^*q^{-2s})(\zeta^2+\zeta^{-2}) - (\tb q^{2s}+\tc q^{-2s})(q^{2s+1} + q^{-2s-1}) \right)\ \label{scpar3}
\eeqa
where $\zeta \in {\mathbb C}^*$ is generic. Adapting the notations from \cite[Theorem 5.3]{T04} with \cite[Lemma 10.3]{Ter04} compared with (\ref{scpar1})-(\ref{scpar3}), the coefficients $\{A^*_{M,M'}\}$,  $\{A_{N,N'}\}$ in (\ref{tridAstar}),  (\ref{tridAstar2}) read as follows: 
\beqa
A^*_{M,M-1}&=&   q^{2-4s} \frac{ (1-q^{2M})(\tc-\tb q^{2M+4s}) (\tb^*q^{2s-1}\zeta^{-2} +\tb q^{2M-2})(   \tc q^{2s-1}\zeta^2 + \tc^*q^{2M-2}) }{(\tc-\tb q^{4M-2})(\tc-\tb q^{4M})}  \ ,\label{amm1}\\
A^*_{M-1,M}&=&  \frac{ (1-q^{2M-4s-2})(\tc- \tb q^{2M-2}) (\tc+\tb^*\zeta^{-2}q^{2M+2s-1})  (\tc^*+\tb\zeta^{2}q^{2M+2s-1})}{(\tc-\tb q^{4M-4})(\tc-\tb q^{4M-2})} \ ,\label{amm2}\\
A^*_{M,M}&=& \theta^*_0 -  A^*_{M,M+1} - A^*_{M,M-1} \ .\label{amm3} \ 
\eeqa
The coefficients $\{A_{N,N'}\}$ are obtained from $\{A^*_{M,M'}\}$ using the transformation $(\tb,\tc,\zeta,M) \leftrightarrow (\tb^*,\tc^*,\zeta^{-1},N) $.\vspace{1mm}

In addition to the Askey-Wilson relations (\ref{aw1}),  (\ref{aw2}), by the Cayley-Hamilton theorem one has:
\beqa
\prod_{M=0}^{2s}\left( \pi(\tA) - \theta_M\right) = 0\ , \quad \prod_{N=0}^{2s}\left( \pi(\tA^*) - \theta^*_N\right) = 0
%\ ,\quad \prod_{M=0}^{2s}\left( \tC - \theta^\diamond_M\right) = 0
\ .\label{polyc}
\eeqa

Note that explicit examples of Leonard pairs associated with embeddings of the AW algebra into  $U_q(sl_2)$ or   $(U_q(sl_2))^{\otimes 3 }$ are known, see \cite{GZ2,GZ93b,Huang}.\vspace{1mm}

\subsection{Transition matrix and the $q$-Racah polynomials}
Given a Leonard pair, the transition matrices relating the eigenbases  $\{|\theta_M \rangle \}_{M=0}^{2s}$ and  $\{|\theta^*_N \rangle \}_{N=0}^{2s}$ are expressed in terms of $q$-Racah polynomials \cite{Z91,Ter04}. Consider
\beqa
  |\theta^*_N\rangle = \sum_{M=0}^{2s} P_{MN}  |\theta_M \rangle \qquad \mbox{and} \qquad |\theta_M\rangle = \sum_{N=0}^{2s} %\langle \theta^*_N |\theta_M \rangle 
	(P^{-1})_{NM}|\theta^*_N \rangle \label{transeq}
\eeqa 
where  $P$ (resp. $P^{-1}$) %$P_{MN}=\langle \theta_M |\theta^*_N \rangle$ (resp. $(P^{-1})_{NM}=\langle \theta^*_N |\theta_M \rangle$) 
%$\{\langle \theta^*_N |\theta_M \rangle\}$ 
denotes the transition matrix from the basis $\{|\theta_M \rangle \}_{M=0}^{2s}$ to the basis $\{|\theta^*_N \rangle \}_{N=0}^{2s}$ (resp. the inverse transition matrix from the basis $\{|\theta^*_N \rangle \}_{N=0}^{2s}$ to the basis $\{|\theta_M \rangle \}_{M=0}^{2s}$). From (\ref{transeq}), using (\ref{orthcond}), the entries of the transition matrix are given by the scalar products:
%\pascal{if the central result of the paper is the transition matrix coeff in terms of Bethe states - i think it is - then it will be better to put lemma below, and erase prop for the q-racah formula - as it is a corollary}
%
\beqa
P_{MN}=\langle \theta_M |\theta^*_N \rangle/\langle \theta_M |\theta_M \rangle \qquad \mbox{and} \qquad (P^{-1})_{NM}=\langle \theta^*_N |\theta_M \rangle/ \langle \theta^*_N |\theta^*_N \rangle\ .\label{transcal}
\eeqa

Similarly, the dual eigenvectors are related as follows:
\beqa
  \langle \theta^*_N| = \sum_{M=0}^{2s} \frac{\xi^*_N}{\xi_M} P^{-1}_{NM}  \langle \theta_M | \qquad \mbox{and} \qquad \langle \theta_M | = \sum_{N=0}^{2s} %\langle \theta^*_N |\theta_M \rangle 
	\frac{\xi_M}{\xi^*_N} P_{MN} \langle \theta^*_N | \label{transeqdual}\ .
\eeqa 
%
%where  
%
%\beqa
%{\tilde P}_{MN}= \qquad \mbox{and} \qquad  {\tilde P}^{-1}_{NM}= \frac{\xi_M}{\xi^*_N} P_{MN}\ .\label{Ptilde}
%\eeqa
%
%It is known that the entries of the transition matrix  are expressed in terms of $q$-Racah polynomials
Introduce the $q$-Racah polynomials:
\beqa
R_M(\theta^*_N)= \fpt{  q^{-2M}  }{  \frac{\tb}{\tc}q^{2M}  }{  q^{-2N}  }{  \frac{\tb^*}{\tc^*}q^{2N}  }{  -\frac{\tb}{\tc^*}q^{2s+1} \zeta^{2}  }{  -\frac{\tb^*}{\tc}q^{2s+1} \zeta^{-2}  }{  q^{-4s}  }\ . \label{qracahpoly}
\eeqa
Adapting the notations of \cite{Ter04}, the entries of the transition matrices are given by:
\beqa
P_{MN}=k_N R_M(\theta^*_N) \qquad \mbox{and}  \qquad (P^{-1})_{NM} =   \nu_0^{-1} k^*_M R_M(\theta^*_N) \label{PMN}
\eeqa
where
\beqa
k_N= \frac{ (-\frac{\tb^*}{\tc}q^{2s+1}\zeta^{-2},-\frac{\tb}{\tc^*}q^{2s+1}\zeta^{2},\frac{\tb^*}{\tc^*},q^{-4s};q^2)_N }{ ( q^2, -\frac{\tb^*}{\tb}q^{1-2s}\zeta^{-2},  -\frac{\tc}{\tc^*}q^{1-2s}\zeta^{2}, \frac{\tb^*}{\tc^*}q^{4s+2} ;q^2)_N} \frac{(1-\frac{\tb^*}{\tc^*}q^{4N})}{(\frac{\tb}{\tc})^N(1-\frac{\tb^*}{\tc^*})}      \ ,\qquad k^*_M= k_N|_{N\rightarrow M, \tb\leftrightarrow \tb^*,\tc\leftrightarrow \tc^*, \zeta \rightarrow \zeta^{-1}} \ , \nonumber
\eeqa
\beqa
 \nu_0=   \frac{ (\frac{\tb}{\tc}q^2, \frac{\tb^*}{\tc^*}q^{2} ;q^2)_{2s}}{(-\frac{\tb^*}{\tc}q^{2s+1}\zeta^{-2})^{2s} ( -\frac{\tb}{\tb^*}q^{1-2s}\zeta^2,  -\frac{\tc}{\tc^*}q^{1-2s}\zeta^{2};q^2)_{2s}}\ .\nonumber  
\eeqa

From (\ref{transeq}), using (\ref{tridAstar}), (\ref{tridAstar2}) one finds that the transition matrix coefficients satisfy three-term recurrence relations with respect to $M,N$ % and a second-order difference equation 
\cite{Z91,Ter04}. Using (\ref{PMN}), one recovers the well-known relations:
\beqa
\theta_N^*  R_M(\theta^*_N)     &=& A^*_{M,M+1} R_{M+1} (\theta^*_N)   +  A^*_{M,M} R_{M} (\theta^*_N)   + A^*_{M,M-1} R_{M-1} (\theta^*_N) \ ,\label{rec}\\
\theta_M  R_M(\theta^*_N)     &=& A_{N,N+1} R_{M} (\theta^*_{N+1})     +  A_{N,N} R_{M} (\theta^*_N)   + A_{N,N-1} R_{M} (\theta^*_{N-1}) \ .\label{qdiff}
\eeqa
%
%A comparison with \cite[Section 3.2]{KS} leads to the following identification:
%
%\beqa
%R_M(\theta^*_N)= \fpt{  q^{-2M}  }{  \frac{\tb}{\tc}q^{2M}  }{  q^{-2N}  }{  \frac{\tb^*}{\tc^*}q^{2N}  }{  -\frac{\tb}{\tc^*}q^{2s+1} \zeta^{2}  }{  -\frac{\tb^*}{\tc}q^{2s+1} \zeta^{-2}  }{  q^{-4s}  }\ . \label{qracahpoly}
%\eeqa
%
%In the literature,  $ R_M(\theta^*_N)$ is the so-called q-Racah polynomial. 
Combining the two relations in (\ref{transeq}), the orthogonality relations satisfied by the $q$-Racah polynomials follow \cite[Section 16]{T03}:
\beqa
\sum_{N=0}^{2s} k_N R_M(\theta^*_N)R_{M'}(\theta^*_N) = \nu_0(k^*_M)^{-1} \delta_{MM'}\ , \qquad \sum_{M=0}^{2s} k_M^* R_M(\theta^*_N)R_{M}(\theta^*_{N'}) = \nu_0(k_{N})^{-1} \delta_{NN'}\ .
%\sum_{N=0}^{2s} \frac{  (-\frac{\tb^*}{\tc}q^{2s+1}\zeta^{-2},-\frac{\tb}{\tc^*}q^{2s+1}\zeta^{2},\frac{\tb^*}{\tc^*},q^{-4s};q^2)_N}{ ( q^2, -\frac{\tc}{\tc^*}q^{1-2s}\zeta^2,  -\frac{\tb^*}{\tb}q^{1-2s}\zeta^{-2}, \frac{\tb^*}{\tc^*}q^{4s+2} ;q^2)_N} \frac{(1+\frac{\tb^*}{\tc^*}q^{4N})}{(\frac{\tb}{\tc})^N(1-\frac{\tb^*}{\tc^*})}R_M(\theta^*_N)R_{M'}(\theta^*_N) = h_M \delta_{MM'}\ \nonumber
\eeqa

From (\ref{transcal}), it turns out that the scalar products between eigenvectors of Leonard pairs are the basic building blocks for the $q$-Racah polynomials. As $R_0(\theta^*_N)=1$ and $R_M(\theta^*_0)=1$, from (\ref{PMN}) one gets $k_N=\langle \theta_0 |\theta^*_N \rangle/\langle \theta_0 |\theta_0 \rangle $ and $\nu_0^{-1}k^*_M=\langle \theta^*_0 |\theta_M \rangle/\langle \theta^*_0 |\theta^*_0 \rangle$. The following result can be viewed as a variation of \cite{Z91}. Up to normalization, see also \cite[Theorem 14.6 and 15.6]{T03}.
For any $0\leq N,M \leq 2s$, the $q$-Racah polynomials are given by:
\beqa
R_M(\theta^*_N) &=& \frac{\langle \theta_M |\theta_N^* \rangle}{\langle \theta_0 |\theta^*_N \rangle}\frac{\langle \theta_0 |\theta_0 \rangle}{\langle \theta_M |\theta_M \rangle}\label{Rac1}\\
&=& \frac{\langle \theta^*_N |\theta_M \rangle}{\langle \theta^*_0 |\theta_M \rangle}\frac{\langle \theta^*_0 |\theta_0^* \rangle}{\langle \theta^*_N |\theta^*_N \rangle}\  .\label{Rac2}
\eeqa

\vspace{2mm}

\section{Eigenbases for Leonard pairs from Bethe states}
The AW algebra admits a presentation in the form of a reflection algebra \cite{Za95,Bas04}. This presentation allows one to apply the technique of algebraic Bethe ansatz in order to diagonalize $\pi({\textsf A}), \pi({\textsf A}^*)$. The purpose of this section is to recall the eigenbases - see \cite{BP19} for details - and construct the dual eigenbases associated with the Leonard pair $\pi({\textsf A}), \pi({\textsf A}^*)$ in terms of the so-called Bethe states
%\footnote{Although part of the results can be found in \cite{BP19}, they are recalled here for self-concistancy.} 
and dual Bethe states. The Bethe eigenstates and dual eigenstates built  are essentially of two different types, called either of {\it homogeneous} type or of {\it inhomogeneous} type. The necessary material for the analysis below is found in Appendix \ref{apA} and \cite{BP19}.
% So,  the labels $`(h)'$ and $`(i)'$ will be used below to distinguish between the two cases. 

\subsection{Bethe  states and dual Bethe states}
In the algebraic Bethe ansatz approach, the main ingredients  are
%is the formulation of  ${\textsf A}, {\textsf A}^*$ in terms of
the `dynamical' operators \cite{gauge} $\{\mathscr{A}^{\epsilon}(u,m), \mathscr{B}^{\epsilon}(u,m),\mathscr{C}^{\epsilon}(u,m),\mathscr{D}^{\epsilon}(u,m)\}$ that satisfy a set of exchange relations (\ref{comBdBd})-(\ref{comDdCd}). Importantly, the dynamical operators  are polynomials of maximal degree $2$ in the elements $\tA,\tA^*$, depending on $\alpha,\beta\in {\mathbb C}$ and the so-called spectral parameter $u\in {\mathbb C}^*$. We refer the reader to \cite[Appendix A]{BP19} for their explicit expressions. 
 \vspace{1mm}

The starting point of the construction of Bethe states is the identification of so-called reference states. 
Given a Leonard pair with (\ref{tridAstar}), (\ref{tridAstar2}), the following lemma is derived \cite[Propositions 3.1, 3.2]{BP19}. Let $m_0$ be an integer. 
\begin{lem}\label{lem:g1} If the parameter $\alpha$ is such that:
\beqa
\mbox{$(q^2-q^{-2})\chi^{-1}\alpha \tc^*q^{m_0}=1$ \qquad (resp. $(q^2-q^{-2})\chi^{-1}\alpha \tb q^{-m_0}=-1$)}\label{ab}
\eeqa
then  %$|\theta^*_0\rangle$ %$|\Omega^+\rangle$
 %(resp. $|\theta_0\rangle$ %$|\Omega^-\rangle$)   satisfies
%
\beqa
\pi(\mathscr{C}^+(u,m_0))%|\Omega^+\rangle
|\theta^*_0\rangle =0\, \qquad \mbox{(resp.  $\pi(\mathscr{C}^-(u,m_0))
%|\Omega^-\rangle
|\theta_0\rangle  =0\,$)}.\label{cmO}
\eeqa
\end{lem}
In the dual vector space, for (\ref{tridAstardual}), (\ref{tridAstar2dual}) and using \cite[Appendix A]{BP19}, analog results are derived along the same line.
\begin{lem}\label{lem:g2}  If the parameter $\beta$ is such that:
\beqa
\mbox{$(q^2-q^{-2})\chi^{-1}\beta \tb^*q^{-m_0+2}=1$ \qquad (resp. $(q^2-q^{-2})\chi^{-1}\beta \tc q^{m_0-2}=-1$)}\label{abdual}
\eeqa
then
\beqa
\langle\theta_0^*| \pi(\mathscr{B}^+(u,m_0-2)) =0\, \qquad \mbox{(resp.
$\langle\theta_0| \pi(\mathscr{B}^-(u,m_0-2))  =0\,$)}.\label{cmO}
\eeqa
\end{lem}

Within the algebraic Bethe ansatz framework, the fundamental eigenvectors for the Leonard pair and their duals find a natural interpretation. 
Let $|\Omega^\pm\rangle$ and  $\langle \Omega^\pm|$ denote the so-called reference and dual reference states with respect to the dynamical operators $\{\mathscr{B}^{\epsilon}(u,m),\mathscr{C}^{\epsilon}(u,m)\}$. 
\begin{defn}\label{defrefst}
\beqa
&&|\theta_0\rangle = |\Omega^-\rangle  \ , \quad   |\theta^{*}_0\rangle = |\Omega^+\rangle \nonumber \ ,\qquad
\langle\theta_0| = \langle\Omega^-| \ , \quad  \langle \theta^{*}_0|= \langle\Omega^+|   \ .
\eeqa
\end{defn}
\vspace{1mm}

According to the choice of parameters $\alpha,\beta$, the action of the dynamical operators $\mathscr{A}^\pm(u,{m_0})$ and $\mathscr{D}^\pm(u,{m_0})$ on the reference states $|\Omega^\pm\rangle$ and duals $\langle \Omega^\pm|$ are computed. Recall the parametrization  (\ref{scpar1})-(\ref{scpar3}) and define:
The vector space $\cal V$ and its dual $\tilde{\cal V}$ being finite dimensional, the following actions
of dynamical operators are considered. Recall (\ref{tridAstar})-(\ref{tridAstar2dual}). The following result extends \cite[Lemma 3.4]{BP19}, thus we skip the proof.
\begin{lem} The following holds: 
%actions of the dynamical operators on the vectors $|\theta_{2s}\rangle $, $|\theta^*_{2s}\rangle $ and
%the associated duals are given by,
%
\beqa
\pi(\mathscr{B}^{+}(u,m_0+4s))|\theta^*_{2s}\rangle  =0\ \quad \mbox{for}\quad\  (q^2-q^{-2})\chi^{-1}\beta \tb^* q^{-m_0}=1\ ,
\eeqa
\beqa
\pi(\mathscr{B}^{-}(u,m_0+4s))|\theta_{2s}\rangle  =0\ \quad \mbox{for}\quad\  (q^2-q^{-2})\chi^{-1}\beta \tc q^{m_0}=-1\ ,
\eeqa
\beqa
\langle\theta^*_{2s}|\pi(\mathscr{C}^{+}(u,m_0+4s))  =0\ \quad \mbox{for}\quad\  (q^2-q^{-2})\chi^{-1}\alpha \tc^* q^{m_0}=1\ ,
\eeqa
\beqa
\langle\theta_{2s}|\pi(\mathscr{C}^{-}(u,m_0+4s))  =0\ \quad \mbox{for}\quad\  (q^2-q^{-2})\chi^{-1}\alpha \tb q^{-m_0}=-1\ .
\eeqa
\end{lem}
By straightforward calculations, it follows (see \cite[Lemma 3.3]{BP19} for the proof of (\ref{actionADvac})):
\begin{lem}\label{lem:diagonalaction} Let $\alpha,\beta$ be fixed according to Lemmas \ref{lem:g1}, \ref{lem:g2}. 
Then, the dynamical operators act as:
\ben
\pi(\mathscr{A}^\pm(u,{m_0}))|\Omega^\pm\rangle &=&\Lambda_1^\pm(u) |\Omega^\pm\rangle \qquad \mbox{and}\qquad 
\pi(\mathscr{D}^\pm(u,m_0))|\Omega^-\rangle =\Lambda_2^\pm(u) |\Omega^\pm\rangle \ ,\label{actionADvac}\\
\langle\Omega^\pm|\pi(\mathscr{A}^\pm(v,{m_0}))&=&\langle\Omega^\pm|\Lambda_1^\pm(v)\qquad \mbox{and}\qquad 
\langle\Omega^\pm|\pi(\mathscr{D}^\pm(v,m_0)) =\langle\Omega^\pm|\Lambda_2^\pm(v)  \ ,\label{actionADvacdual}
\een
where the eigenvalues  take the factorized form:
\ben\label{Lap}
&&\Lambda_1^\epsilon(u)=\frac{q^{-2s-1}}{u^{\epsilon}}
\left( q^{2 s+1} u\zeta^{-1}-u^{-1}\zeta\right)\left( q^{2 s+1}u\zeta-u^{-1}\zeta^{-1}\right) 
\nonumber
\\ && \qquad\qquad\qquad
\times
\left(u \tc^* q^{-2s}  +u^{-1} \tb q^{2s}\right)
\left(
u \left(\frac{\tc}{\tc^*}\right)^{\frac{1-\epsilon}{2}} +
u^{-1} \left(\frac{\tb^*}{\tb}\right)^{\frac{1+\epsilon}{2}}
\right)\,,\nonumber
\een
\ben
&& \Lambda_2^\epsilon(u)=
\frac{(u^2-u^{-2})q^{-2s-1}}{u^\epsilon (qu^2-q^{-1}u^{-2})}
\left( q^{2 s-1} u^{-1}\zeta-u \zeta^{-1}\right)
\left(q^{2 s-1} u^{-1} \zeta^{-1}-u \zeta\right) \nonumber
\\ && \qquad\qquad\qquad
\times
\left(q^2 u \tb q^{2s}+u^{-1} \tc^* q^{-2s}\right)
\left(q^2 u \left(\frac{\tb^*}{\tb}\right)^{\frac{1+\epsilon}{2}}+u^{-1} 
\left(\frac{\tc}{\tc^*}\right)^{\frac{1-\epsilon}{2}}\right)\,.\nonumber
\een 
\end{lem}
Different types of  Bethe states may be considered, built from successive actions of the dynamical operators $\mathscr{B}^\pm(u,m)$ (resp. $\mathscr{C}^\pm(u,m)$) on each reference state
$|\Omega^\pm\rangle$ (resp. its dual $\langle\Omega^\pm|$). Consider the strings of dynamical operators 
\ben
B^{\epsilon}(\bar u,m,M)&=&\mathscr{B}^{\epsilon}(u_1,m+2(M-1))\cdots \mathscr{B}^{\epsilon}(u_M,m)\, \label{SB}
\een
and
\ben
C^{\epsilon}(\bar v,m,N)&=&\mathscr{C}^{\epsilon}(v_1,m+2)\cdots \mathscr{C}^{\epsilon}(v_N,m+2N)\, ,\label{stringC1}
\een
where we denote the set of variables $ \bar  u = \{u_1,u_2,\dots,u_M\}$, $ \bar v= \{v_1,v_2,\dots,v_N\}$. Taking into account the parameters $\alpha$ and $\beta$ given in (\ref{ab}) and (\ref{abdual}), define the following vectors and dual vectors,
\ben\label{PsiA}
|\Psi_{-}^M( \bar u,m_0)\rangle = \pi(B^{-}( \bar  u,m_0,M))|\Omega^{-}\rangle\,\quad
\mbox{for} \quad (q^2-q^{-2})\chi^{-1}\alpha \tb q^{-m_0}=-1 \quad \mbox{and}\quad \beta=0\ ,
\een
\ben\label{PsiAp}
|\Psi_{+}^M( \bar  w,m_0)\rangle = \pi(B^{+}( \bar  w,m_0,M))|\Omega^{+}\rangle\,\quad
\mbox{for} \quad (q^2-q^{-2})\chi^{-1}\alpha \tc^*q^{m_0}=1 \quad \mbox{and}\quad \beta=0\ ,
\een
\ben\label{PsidualAstarm}
\langle \Psi_{-}^N( \bar  v,m_0)| = \langle \Omega^{-}| \pi(C^{-}( \bar  v,m_0,N))\,
\quad
\mbox{for} \quad (q^2-q^{-2})\chi^{-1}\beta \tc q^{m_0-2}=-1 \quad \mbox{and}\quad \alpha=0\ ,
\een
\ben\label{PsidualAstar}
\langle \Psi_{+}^N( \bar  y,m_0)| = \langle \Omega^{+}| \pi(C^{+}( \bar  y,m_0,N))\,
\quad
\mbox{for} \quad (q^2-q^{-2})\chi^{-1}\beta \tb^*q^{-m_0+2}=1 \quad \mbox{and}\quad \alpha=0\ .
\een

As usual, if $ \bar u$ (or $ \bar w, \bar v, \bar y$) is a set of variables satisfying certain Bethe ansatz equations, the Bethe states (\ref{PsiA}), (\ref{PsiAp}), (\ref{PsidualAstarm}) and (\ref{PsidualAstar}), are called `on-shell'. On the other hand, 
if the set of variables $ \bar  u$ is arbitrary, the Bethe states are called `off-shell'.

\subsection{Eigenbases of homogenous type for the Leonard pair % $\pi({\textsf A}), \pi({\textsf A}^*)$
}
The eigenvectors and dual eigenvectors for the Leonard pair $\pi({\textsf A})$, $\pi({\textsf A}^*)$, are now constructed in terms of Bethe states and dual Bethe states associated with Bethe equations of homogeneous type. 
Let us define the set of functions:
\ben
E_{\pm}^M(u_i, \bar u_i)=-\frac{b(u_i^2)}{b(qu_i^2)}\prod_{j=1,j\neq i}^Mf(u_i,u_j)\Lambda_1^\pm(u_i)+\prod_{j=1,j\neq i}^Mh(u_i,u_j)\Lambda_2^\pm(u_i)\,,\label{Bfunc}
\een
for $i=1,\dots,M$. The set of equations $E_{\pm}^M(u_i, \bar  u_i)=0$ for $i=1,\dots,M$
are called the Bethe ansatz equations of {\it homogeneous} type associated with the set of Bethe roots ${\bar u}$. 
Note that the set of equations $E_{\pm}^M(u_i, \bar  u_i)=0$
for $i=1,\dots,M$
contain \textit{trivial} solutions where $u_i^2=u_i^{-2}$ or $u_i=0$ for $i=1,\dots,M$, recall the expression of $\Lambda_2^\epsilon(u)$ in Lemma \ref{lem:diagonalaction}. These solutions must be discarded since they lead to identically null or ill defined
Bethe states. To see that, we refer the reader to \cite[Corollary 3.1]{BP19}
where the expansion of a Bethe state in the Poincar\'e-Birkhoff-Witt basis of the Askey-Wilson algebra is given. Also,
after extracting the denominator in $E_{\pm}^M(u_i, \bar  u_i)=0$, we may find solutions for
which $U_i=U_j$ for $i\neq j$ where the symmetrized Bethe root $U_i=(qu_i^2+q^{-1}u_i^{-2})/(q+q^{-1})$ is introduced. The solutions
with coincident symmetrized Bethe roots such that $U_i=U_j$ for $i\neq j$ are not admissible, since they may lead to Bethe states which are not eigenstates of $\pi(\tA)$ or $\pi(\tA^*)$. For additional discussion on this subject we refer the reader to
\cite[Subsection 3.5]{BP19}. Below, the solutions $U_i\neq U_j$ are called \textit{admissible}. For a discussion of coincident Bethe roots for other integrable systems
see \cite[Subsection 2.3]{S22}. Based on these observations and supported by numerical analysis, we formulate
the following hypothesis.
\vspace{1mm}

\begin{hyp}\label{hyp1}
For each integer $M$ (resp. $N$) with $0\leq M,N\leq 2s$,  there exists at least one set of non trivial admissible Bethe roots $S^{M(h)}_-=\{u_1,...,u_{M}\}$ (resp. $S_+^{*N(h)}=\{w_1,...,w_{N}\}$) such that
\beqa
 E_{-}^M(u_i, \bar u_i)=0 \quad \mbox{for} \quad  \bar u = S^{M(h)}_- \ ,\qquad
(\mbox{resp.}\quad  E_{+}^N(w_i, \bar  w_i)=0 \quad \mbox{for} \quad \bar  w = S^{*N(h)}_+)\ .
\eeqa
%
%associated with the eigenvalue $\theta_M^{(h)}$ (resp. $\theta_N^{*(h)}$) of the form (\ref{st}).   
%
\end{hyp}
\begin{lem}\label{lem:tM}  Assume Hypothesis \ref{hyp1}. The following relations hold: %The Bethe eigenstates  and dual eigenstates of the Leonard pair $\bar\pi(\tA)$, $\bar\pi(\tA^*)$, defined as:
\beqa
|\theta_M\rangle &=& {\cal N}_M(\bar u) |\Psi_{-}^M(\bar u,m_0)\rangle \qquad \mbox{for}\quad  \bar u = S^{M(h)}_-\ ,\label{norm1}\\
|\theta_N^{*}\rangle &=& {\cal N}_N^*(\bar w) |\Psi_{+}^N(\bar w,m_0)\rangle \qquad \mbox{for}  \quad  \bar w = S^{*N(h)}_+\ 
\label{norm2}
\eeqa
with
\beqa
{\cal N}_M(\bar u)= \prod_{k=1}^M\left(q u_k b(u_k^2)A^*_{k,k-1}\right)^{-1} \ ,\qquad
{\cal N}_N^*(\bar w)= \prod_{k=1}^N\left(-q^{-1} w_k^{-1} b(w_k^2)A_{k,k-1}\right)^{-1}   \ ,\label{Ncoeff}
\eeqa
and ${\cal N}_0(.)={\cal N}_0^*(.)=1$.
\end{lem}
\begin{proof} Consider (\ref{norm1}). By \cite[Proposition 3.1]{BP19}, it is known that:
\beqa
\pi(\tA) |\Psi_{-}^M( \bar u,m_0)\rangle &=&  \theta_M  |\Psi_{-}^M( \bar u,m_0)\rangle \qquad \mbox{for} \quad  \bar u = S^{M(h)}_- \ .\label{Apsi}
\eeqa
By definition of a Leonard pair, the spectrum of $\pi(\tA)$ is non-degenerate with (\ref{tridAstar}). So, if there exists a solution of the Bethe equations associated with the eigenvalue $\theta_M$, it must be such that $|\theta_M\rangle $ is proportional to $|\Psi_{-}^M( \bar u,m_0)\rangle$. Let ${\cal N}_M(\bar u)$ denote the normalization factor in the r.h.s. of (\ref{norm1}). To fix it, observe that $B^{-}(\bar u,m_0,M)$ is a polynomial in $\tA,\tA^*,\tA\tA^*,\tA^*\tA$ \cite[Appendix A]{BP19}. Using (\ref{tridAstar}), one extracts the coefficient of $|\theta_M\rangle$ from $|\Psi_{-}^M( \bar u,m_0)\rangle$ which, by definition, is the inverse of ${\cal N}_M(\bar u)$. The proof of (\ref{norm2}) is done along the same line, starting from \cite[Proposition 3.2]{BP19} and using (\ref{tridAstar2}). 
\end{proof}

The proof of the following lemma is analog to Lemma \ref{lem:tM}, so we skip the details.
\begin{lem}\label{lem:tMdual}  Assume Hypothesis \ref{hyp1}.  The following relations hold: 
\beqa
\langle \theta_M|&=&\tilde{\cal N}_M(\bar v)\langle \Psi_{-}^M(\bar v,m_0)| \qquad \mbox{for}  \quad  \bar v = S^{M(h)}_-\ ,\label{normdual1}\\
\langle \theta_N^{*}|&=&\tilde{\cal N}_N^{*}(\bar y)\langle \Psi_{+}^N(\bar y,m_0)| \qquad \mbox{for}  \quad  \bar y = S^{*N(h)}_+ \label{normdual2}
\eeqa
with 
\beqa
\tilde{\cal N}_M(\bar v)&=&\prod_{k=1}^M\left(q^{-1} v_k b(v_k^2){\tilde A}_{k,k-1}^*\right)^{-1}    \ ,\qquad
\tilde{\cal N}_N^{*}(\bar y)=\prod_{k=1}^N\left(-q y_k^{-1} b(y_k^2){\tilde A}_{k,k-1}\right)^{-1}    \ \label{Ncoeffdual}
\eeqa
and $\tilde{\cal N}_0(.)=\tilde{\cal N}_0^*(.)=1$.
\end{lem}
For the homogeneous case, it should be stressed that the construction of the Bethe states and dual Bethe states lead to the same set of Bethe equations (\ref{Bfunc}).

\subsection{Eigenbases of inhomogenous type for the Leonard pair %$\bar\pi({\textsf A}), \bar\pi({\textsf A}^*)$
}
The eigenvectors and dual eigenvectors for the Leonard pair $\pi({\textsf A})$, $\pi({\textsf A}^*)$, can be alternatively constructed in terms of Bethe states and dual Bethe states associated with Bethe equations of inhomogeneous type. 
However, contrary to the homogeneous case, for the inhomogeneous case two different sets of Bethe equations characterizing respectively to Bethe states and dual Bethe states are obtained.
\vspace{1mm}

A first set of Bethe equations of inhomogeneous type is now introduced. Let us define\footnote{We observe
that there is a typo in Eq.(3.78) of \cite{BP19}: the second and third terms in Eq.(3.78) have wrong sign.}:
\ben
E_{\pm}(u_i,\bar u_i)&=& \frac{b(u_i^2)}{b(qu_i^2)}u_i^{\pm 1}\prod_{j=1,j\neq i}^{2s}f(u_i,u_j)\Lambda_1^{\pm 1}(u_i)
-
(q^2u_i^3)^{\mp 1}\prod_{j=1,j\neq i}^{2s}h(u_i,u_j)\Lambda_2^{\pm 1}(u_i)
\label{BAEd}\\
&&
\  +    \nu_\pm
\frac{u_i^{\mp 2}b(u_i^2)}{b(q)}\frac{\prod_{k=0}^{2s}b(q^{1/2+k-s}\zeta u_i)b(q^{1/2+k-s}\zeta^{-1}u_i)}{\prod_{j=1,j\neq i}^{2s}b(u_iu_j^{-1}) b(qu_iu_j)}\, =0\ ,\nonumber
\, \label{Bfuncinhom}
\een
where 
\ben
\nu_+ = q^{-1-4s}\tc^*\,,\quad
\nu_- = q^{1+4s}\tb\,,
\label{deltad}
\een
for $i=1,\dots,2s$. The set of equations  $E_{\pm}(u_i,\bar u_i)=0$ for $i=1,\dots,2s$
are the Bethe ansatz equations of {\it inhomogeneous} type for the set of Bethe roots $\bar u$ with $M=2s$. 
%and analogously for any set $\bar w$, $\bar v$ or $\bar y$.
\vspace{1mm}

Recall the structure of the spectra (\ref{st}) for a Leonard pair.  
\begin{hyp}\label{hyp2}
For each integer $M$ (or $N$) with $0\leq M,N\leq 2s$,  there exists at least one set of non trivial admissible Bethe roots $S_+^{M(i)}=\{u_1,...,u_{2s}\}$ (resp. $S_-^{*N(i)}=\{w_1,...,w_{2s}\}$) such that 
\beqa
 E_{+}(u_j,\bar u_j)=0 \quad \mbox{for} \quad \bar u = S^{M(i)}_+ 
\ \qquad
(\mbox{resp.}\quad  E_{-}(w_j,\bar w_j)=0 
\quad \mbox{for} \quad \bar w = S^{*N(i)}_-)
\ ,
\eeqa
and associated with the following equality 
\beqa
\theta_M&=&  q^{-4s} \Big( \tc^* (\zeta^2+\zeta^{-2})[2s]_q +
q^{2s}(\tb q^{2s}+\tc q^{-2s})
 -q\tc^*\sum_{j=1}^{2s} (qu_j^2+q^{-1}u_j^{-2})\Big) \quad \mbox{for} \quad \bar u = S^{M(i)}_+\ \label{eigd1}\\
 \qquad \quad  \mbox{(resp.} \ \  \theta_N^{*}&=& q^{4s}\Big( \tb  (\zeta^2+\zeta^{-2})[2s]_q +
q^{-2s}(\tb^* q^{2s}+\tc^* q^{-2s})
-q^{-1}\tb\sum_{j=1}^{2s} (qw_j^2+q^{-1}w_j^{-2}) \Big) \ \quad \label{eigd2}\\ \nonumber&\mbox{for}& \quad \bar w = S^{*N(i)}_-\ )\ . 
\eeqa
\end{hyp}

The equality (\ref{eigd1}) is proven in the next lemma. The equality (\ref{eigd2}) is proven along the same line. For numerical examples of (\ref{eigd1}), see \cite[eq. (4.6) and Table 1]{BP19}. \vspace{1mm}

\begin{lem}\label{lem:tMi} Assume Hypothesis \ref{hyp2}. The following relations hold: %The Bethe eigenstates  and dual eigenstates of the Leonard pair $\bar\pi(\tA)$, $\bar\pi(\tA^*)$, defined as:
\beqa
|\theta_M\rangle &=& {\cal N}^{(i)}_M(\bar u') |\Psi_{+}^{2s}(\bar u',m_0)\rangle \qquad \mbox{for}\quad  {\bar u'}= S_+^{M(i)}\ ,\label{normi1}\\
|\theta_N^{*}\rangle &=& {\cal N}^{*(i)}_N(\bar w') |\Psi_{-}^{2s}(\bar w',m_0)\rangle \qquad \mbox{for}  \quad  {\bar w'}=  S^{*N(i)}_-\ 
\label{normi2}
\eeqa
with
\beqa
{\cal N}^{(i)}_M(\bar u')= {\cal N}_{2s}^{*}(\bar u')(P^{-1})_{2s,M} \ ,\qquad
{\cal N}^{*(i)}_N(\bar w')=  {\cal N}_{2s}(\bar w')P_{2s,N} \ .\label{Ncoeffi}
\eeqa
\end{lem}
\begin{proof} Consider (\ref{normi1}). Specializing \cite[Proposition 3.3]{BP19} for $(\kappa,\kappa^*)=(1,0)$, it follows:
\beqa
\pi(\tA) |\Psi_{+}^{2s}(\bar u',m_0)\rangle &=&  \theta(\bar u')  |\Psi_{+}^{2s}(\bar u',m_0)\rangle \ \quad \mbox{for}\qquad E_{+}(u'_j,\bar u'_j)=0\ , 
\eeqa
where $\theta(\bar u')$ denotes the r.h.s. of (\ref{eigd1}) with $u_j \rightarrow u'_j$. Now, by definition of a Leonard pair and our choice of parameterization of the structure constants (\ref{scpar1})-(\ref{scpar3}), the  eigenvalues of $\pi(\tA)$ are of the form (\ref{st}). So, if there exists a set of solutions $\{u'_1,...,u'_{2s}\}$ of $E_{+}(u'_j,\bar u'_j)=0$ for $j=1,...,2s$, it always exists an integer $M$ such that $\theta_M= \theta(\bar u')$. Let us denote the corresponding set by ${\bar u'}= S_+^{M(i)}$. The absence of degeneracies in the spectrum of $\pi(\tA)$ implies that $|\theta_M\rangle$ is proportional to $|\Psi_{+}^{2s}(\bar u',m_0)\rangle$. To determine the normalization coefficient ${\cal N}^{(i)}_M(\bar u')$ in (\ref{normi1}), one compares
\beqa
{\cal N}^{(i)}_M(\bar u') |\Psi_{+}^{2s}(\bar u',m_0)\rangle = {\cal N}^{(i)}_M(\bar u') \left( {\cal N}^{*}_{2s}(\bar u')^{-1}  |\theta^*_{2s}\rangle + \cdots\right) \ 
\eeqa
with the second equation in (\ref{transeq}): 
\beqa
|\theta_{M}\rangle = (P^{-1})_{2s,M}|\theta^*_{2s}\rangle + \cdots 
\eeqa
The proof of (\ref{normi2}) is done along the same line, starting from the specialization of \cite[Proposition 3.3]{BP19} for $(\kappa,\kappa^*)=(0,1)$.  
\end{proof}

We now turn to the construction of the dual eigenstates for $\bar\pi\left(\tA \right)$, $\bar\pi\left(\tA^* \right)$.  %Recall (\ref{transeqdual}) and (\ref{Ncoeffdual}). 
To this end, a second set of Bethe equations of inhomogeneous type is introduced. Let us define: 
\ben
\tilde{E}_{\pm}(y_i,\bar y_i)&=& \frac{b(y_i^2)}{b(qy_i^2)}y_i^{\mp 1}\prod_{j=1,j\neq i}^{2s}f(y_i,u_j)\Lambda_1^{\pm 1}(y_i)
-
(q^2y_i^3)^{\pm 1}\prod_{j=1,j\neq i}^{2s}h(y_i,u_j)\Lambda_2^{\pm 1}(y_i)
\label{dBAEd}\\
&&
\  +    \tilde{\nu}_\pm
\frac{b(y_i^2)}{b(q)}\frac{\prod_{k=0}^{2s}b(q^{1/2+k-s}\zeta y_i)b(q^{1/2+k-s}\zeta^{-1}y_i)}{\prod_{j=1,j\neq i}^{2s}b(y_iu_j^{-1}) b(qy_iu_j)}\, =0\ ,\nonumber
\, \label{dBfuncinhom}
\een
where
\ben
\tilde\nu_+ = q^{1+4s}\tb^*\,,\quad
\tilde\nu_- = q^{-1-4s}\tc\,,
\label{ddeltad}
\een
for $i=1,\dots,2s$. The set of equations  $\tilde{E}_{\pm}(y_i,\bar y_i)=0$ for $i=1,\dots,2s$
are the `dual' Bethe ansatz equations of {\it inhomogeneous} type for the set of Bethe roots $\bar y$ with $M=2s$.
%and analogously for any set $\bar w$, $\bar v$ or $\bar y$.
%\vspace{1mm}

%
\begin{hyp}\label{hyp3}
For each integer $M$ (or $N$) with $0\leq M,N\leq 2s$,  there exists at least one set of non trivial admissible Bethe roots $dS_+^{M(i)}=\{y_1,...,y_{2s}\}$ (resp. $dS_-^{*N(i)}=\{v_1,...,v_{2s}\}$) such that 
\beqa
 \tilde{E}_{+}(y_j,\bar y_j)=0 \quad \mbox{for} \quad \bar y = dS^{M(i)}_+ 
\ \qquad
(\mbox{resp.}\quad  \tilde{E}_{-}(v_j,\bar v_j)=0 
\quad \mbox{for} \quad \bar v = dS^{*N(i)}_-)
\ ,
\eeqa
and associated with the following equality  
\beqa
\theta_M&=&  q^{4s} \Big( \tb^* (\zeta^2+\zeta^{-2})[2s]_q +
q^{-2s}(\tb q^{2s}+\tc q^{-2s})
-q^{-1}\tb^*\sum_{j=1}^{2s} (qy_j^2+q^{-1}y_j^{-2})\Big) \label{deigd1}\\ \nonumber
\quad &\mbox{for}& \quad \bar y = dS^{M(i)}_+\,.\ \\
\qquad \quad  \mbox{(resp.} \ \
\theta_N^{*}&=& q^{-4s}  \Big(  \tc (\zeta^2+\zeta^{-2})[2s]_q +
q^{2s}(\tb^* q^{2s}+\tc^* q^{-2s})
-q\tc\sum_{j=1}^{2s} (qv_j^2+q^{-1}v_j^{-2}) \Big) \label{deigd2}\\ \nonumber&\mbox{for}& \quad \bar v = dS^{*N(i)}_-\ )\ . 
\eeqa
\end{hyp}
\begin{lem}\label{lem:tMidual}  Assume Hypothesis \ref{hyp3}. The following relations hold: %The Bethe eigenstates  and dual eigenstates of the Leonard pair $\bar\pi(\tA)$, $\bar\pi(\tA^*)$, defined as:
\beqa
\langle \theta_M| &=& \tilde{\cal N}^{(i)}_M(\bar v') \langle\Psi_{+}^{2s}(\bar v',m_0)| \qquad \mbox{for}\quad  {\bar v'}= dS_+^{M(i)}\ ,\label{normduali1}\\
\langle \theta_N^{*}| &=& \tilde{\cal N}^{*(i)}_N(\bar y') \langle\Psi_{-}^{2s}(\bar y',m_0)| \qquad \mbox{for}  \quad  {\bar y'}=  dS^{*N(i)}_-\ 
\label{normduali2}
\eeqa
with
\beqa
\tilde{\cal N}^{(i)}_M(\bar v')= \tilde{\cal N}_{2s}^{*}(\bar v')P_{M,2s}\frac{\xi_M}{\xi^*_{2s}}  \ ,\qquad
\tilde{\cal N}^{*(i)}_N(\bar y')=  \tilde{\cal N}_{2s}(\bar y')(P^{-1})_{N,2s}\frac{\xi^*_N}{\xi_{2s}}  \ .\label{Ncoeffduali}
\eeqa
\end{lem}
\begin{proof} Consider (\ref{normduali1}) and assume Hypothesis \ref{hyp3}. Let us show:
\beqa
\langle \Psi_{+}^{2s}(\bar v',m_0)|\pi(\tA) =  \langle \Psi_{+}^{2s}(\bar v',m_0)| \theta(\bar v') \quad \mbox{for}\quad \bar v'=dS_+^{M(i)}\ ,\label{intA}
 \eeqa
where $\theta(\bar v')$ denotes the r.h.s. of (\ref{deigd1}) with $y_j \rightarrow v'_j$. The proof of (\ref{intA}) follows standard computations within the algebraic Bethe ansatz approach, using the material given in Appendix \ref{apA} together with 
the following relation that holds for all $v$ and  $v_i$, $i=1,...,2s$: 
\beqa
\langle \Psi_{+}^{2s}(\bar v,m_0)|   \bar\pi(\mathscr{C}^+(v,m_0+4s)) &=& \delta \frac{b(v^2)}{v} \frac{ \prod_{k=0}^{2s}b(q^{1/2+k-s}v\zeta)b(q^{1/2+k-s}v\zeta^{-1})}{\prod_{v_i\in \bar v}b(v/v_i)b(qvv_i)}   \langle \Psi_{+}^{2s}(\bar v,m_0)| \nonumber\\
&& - \delta\sum_{v_i\in \bar v} \frac{b(v_i^2)}{v_i} \frac{ \prod_{k=0}^{2s}b(q^{1/2+k-s}v_i\zeta)b(q^{1/2+k-s}v_i\zeta^{-1})}{b(v/v_i)b(qvv_i)  \prod_{j\neq i} b(v_i/v_j)b(qv_iv_j)}    \langle \Psi_{+}^{2s}(\{v,\bar v_i\},m_0)| \ 
\eeqa 
with $\delta = \tb^* q^{4s}$. Note that the proof of this latter relation being analog to \cite[Appendix C]{BP19}, we skip it. Firstly, one obtains:
\beqa
\langle \Psi_{+}^{2s}(\bar v,m_0)| \pi\left(\tA \right) = \langle \Psi_{+}^{2s}(\bar v,m_0)| \lambda_{+}^{2s}(v,{\bar v}) - \frac{vb(q)}{b(v^2)}\sum_{v_i\in \bar v} \frac{ E_{+}(v_i,\bar v_i)}{b(v/v_i)b(qvv_i)}  \langle \Psi_{+}^{2s}(\{v, \bar v_i\},m_0)|\ \label{specdiag}
\eeqa
with 
\beqa
\lambda_{+}^{2s}(v,{\bar v})&=&  \frac{v^{-1}}{b(v^2)b(qv^2)}  \prod_{v_j\in \bar v}f(v,v_j)\Lambda_1^+(v)  + \frac{q^2v^{3}}{b(v^2)b(q^2v^2)}  \prod_{v_j\in \bar v}h(v,v_j)\Lambda_2^+(v) - q\delta  \frac{ \prod_{k=0}^{2s}b(q^{1/2+k-s}v\zeta)b(q^{1/2+k-s}v\zeta^{-1})}{\prod_{v_i\in \bar v}b(v/v_i)b(qvv_i)}
\non\\ &+&  \frac{\left(q\,v\,\bar\eta(v)+q^{-1}v^{-1}\bar\eta(v^{-1})\right)}{b(v^2)b(q^2v^2)}   \      \nonumber
\eeqa
and
\beqa
E_{+}(v_i,\bar v_i) &=& \frac{b(v_i^2)}{v_i b(qv_i^2)}\prod_{v_j\in \bar v_i}f(v_i,v_j)\Lambda_1^+(v_i)
-
q^{2}v_i^{3}\prod_{v_j\in \bar v_i}h(v_i,v_j)\Lambda_2^+(v_i)  \label{BAEd}\\
&& + \delta \frac{qb(v_i^2)}{b(q)}\frac{ \prod_{k=0}^{2s}b(q^{1/2+k-s}v_i\zeta)b(q^{1/2+k-s}v_i\zeta^{-1})}{ \prod_{v_j\in \bar v_i}b(v_i/v_j)b(qv_iv_j)} \ \nonumber
\eeqa
for $i=1,\dots,2s$.  Secondly, by Hypothesis \ref{hyp3}, eq. (\ref{specdiag}) reduces to (\ref{intA}) for $\bar v =\bar v'=dS_+^{M(i)}$, where $\theta(\bar v')= \lambda_{+}^{2s}(v,{\bar v'})$. By studying the singular part of  $\lambda_{+}^{2s}(v,{\bar v'})$, one finds $\theta(\bar v')$ reduces to the r.h.s. of (\ref{deigd1}).  This concludes the proof of (\ref{intA}).   The absence of degeneracies in the spectrum of $\pi(\tA)$ in (\ref{tridAstardual}) implies that $\langle \theta_M|$ is proportional to $\langle \Psi_{+}^{2s}(\bar v',m_0)|$. The normalization coefficient $\tilde{\cal N}^{(i)}_M(\bar v')$ in (\ref{normduali1}) is determined through the comparison between
\beqa
\tilde{\cal N}^{(i)}_M(\bar v') \langle \Psi_{+}^{2s}(\bar v',m_0)| = \tilde{\cal N}^{(i)}_M(\bar v') \left( \tilde{\cal N}^{*}_{2s}(\bar v')^{-1}  \langle \theta^*_{2s}| + \cdots\right) \ 
\eeqa
and the second equation in (\ref{transeqdual}): 
\beqa
\langle\theta_{M}| = P_{M,2s}\frac{\xi_M}{\xi^*_{2s}}\langle \theta^*_{2s}| + \cdots 
\eeqa
The proof of (\ref{normduali2}) is done along the same line.
\end{proof}

\subsection{Relating homogeneous and inhomogeneous Bethe states}
In the algebraic Bethe ansatz framework, in general relating solutions of eigenproblems of homogeneous and inhomogeneous types might appear as a complicated problem. In the  case studied in the present letter, the connection between the eigenbases of Leonard pairs and Bethe eigenstates implies the following identities, straightforward consequences of previous results. From Lemmas \ref{lem:tM}, \ref{lem:tMdual}, \ref{lem:tMi}, \ref{lem:tMidual}, it follows:
\beqa
 {\cal N}_M(\bar u) |\Psi_{-}^M(\bar u,m_0)\rangle    &=&   {\cal N}^{(i)}_M(\bar u') |\Psi_{+}^{2s}(\bar u',m_0)\rangle            \qquad \mbox{for}\quad  \bar u = S^{M(h)}_-\ ,\ {\bar u'}= S_+^{M(i)}\ ,\label{eqBs1}\\
 {\cal N}_N^*(\bar w) |\Psi_{+}^N(\bar w,m_0)\rangle  &=&  {\cal N}^{*(i)}_N(\bar w') |\Psi_{-}^{2s}(\bar w',m_0)\rangle             \qquad \mbox{for}  \quad  \bar w = S^{*N(h)}_+\ ,\ {\bar w'}=  S^{*N(i)}_-\ \label{eqBs2}
\eeqa
and
\beqa
\tilde{\cal N}_M(\bar v)\langle \Psi_{-}^M(\bar v,m_0)| &=&  \tilde{\cal N}^{(i)}_M(\bar v') \langle\Psi_{+}^{2s}(\bar v',m_0)|     \qquad \mbox{for}  \quad  \bar v = S^{M(h)}_-\ ,\ {\bar v'}= dS_+^{M(i)}\ ,\label{eqBsdual1}\\
\tilde{\cal N}_N^{*}(\bar y)\langle \Psi_{+}^N(\bar y,m_0)| &=& \tilde{\cal N}^{*(i)}_N(\bar y') \langle\Psi_{-}^{2s}(\bar y',m_0)|\qquad \mbox{for}  \quad  \bar y = S^{*N(h)}_+ \ ,\ {\bar y'}=  dS^{*N(i)}_-\ .\label{eqBsdual2}
\eeqa

Specializing above identities, in particular the reference states of Definition \ref{defrefst} can be written as inhomogeneous Bethe states:
\beqa
|\Omega^{-}\rangle    &=&   {\cal N}^{(i)}_0(\bar u') |\Psi_{+}^{2s}(\bar u',m_0)\rangle            \qquad \mbox{for}\quad   {\bar u'}= S_+^{0(i)}\ ,\label{eqBs1}\\
|\Omega^{+}\rangle   &=&  {\cal N}^{*(i)}_0(\bar w') |\Psi_{-}^{2s}(\bar w',m_0)\rangle             \qquad \mbox{for}  \quad  {\bar w'}=  S^{*0(i)}_-\ .\label{eqBs2}
\eeqa

\subsection{The $q$-Racah polynomials}\label{sec:qrac}
From the results of the previous sections, the $q$-Racah polynomials can be expressed in terms of ratios of certain scalar products of Bethe states, either of homogeneous or inhomogenous type. Three examples are now displayed.
\begin{prop} The $q$-Racah polynomials are given by:
\beqa
R_M(\theta^*_N) =  {\cal N}^*_N(\bar v)^{-1} \frac{ \langle \Psi_{+}^N(\bar v,m_0)|\Psi_{-}^M(\bar u,m_0)\rangle}{ \langle \Omega^+|\Psi_{-}^M(\bar u,m_0)\rangle} \frac{ \langle \Omega^+|\Omega^+\rangle}{\langle \Psi_{+}^N(\bar v,m_0)|\Psi_{+}^N(\bar v,m_0)\rangle } \  \label{ratioRacah(h)}
\eeqa
for $\bar u=S_-^{M(h)}$, $\bar v=S_+^{*N(h)}$.
\end{prop}
\begin{proof} 
Recall Hypothesis \ref{hyp1}. Consider the expression of the $q$-Racah polynomials given by (\ref{Rac2}). Insert (\ref{norm1}), (\ref{norm2}), (\ref{normdual2}).
\end{proof}
Variations of above expression can be derived using the correspondence between eigenvectors of elements of the Leonard pair and inhomogeneous Bethe states. For instance, recall Hypothesis \ref{hyp1}, \ref{hyp2}. Insert (\ref{normdual1}), (\ref{normi1}), (\ref{normi2})  in (\ref{Rac1}). The $q$-Racah polynomials are now given by:
\beqa
R_M(\theta^*_N) =  \cal N_M(\bar u)^{-1} \frac{ \langle \Psi_{-}^{M}(\bar v,m_0)|\Psi_{-}^{2s}(\bar w,m_0)\rangle}{ \langle \Omega^-|\Psi_{-}^{2s}(\bar w,m_0)\rangle} \frac{ \langle \Omega^-|\Omega^-\rangle}{\langle \Psi_{-}^{M}(\bar v,m_0)|\Psi_{+}^{2s}(\bar u,m_0)\rangle } \  \label{ratioRacah(h)}
\eeqa
for $\bar v=S_-^{M(h)}$, $\bar w=S_-^{*N(i)}$, $\bar u=S_+^{M(i)}$.
Another example is obtained as follows. Insert (\ref{normi1}), (\ref{normi2}), (\ref{normduali2})  in (\ref{Rac2}). One gets:
\beqa
R_M(\theta^*_N) =  \tilde{\cal N}^*_N(\bar y')^{-1} \frac{ \langle \Psi_{-}^{2s}(\bar y',m_0)|\Psi_{+}^{2s}(\bar u,m_0)\rangle}{ \langle \Omega^+|\Psi_{+}^{2s}(\bar u,m_0)\rangle} \frac{ \langle \Omega^+|\Omega^+\rangle}{\langle \Psi_{-}^{2s}(\bar y',m_0)|\Psi_{-}^{2s}(\bar y,m_0)\rangle } \  \label{ratioRacah(h)}
\eeqa
for $\bar u=S_+^{M(i)}$, $\bar y=S_-^{*N(i)}$, $\bar y'=dS_-^{*N(i)}$.

\section{Concluding remarks}
The main result of this letter is a correspondence between the $q$-Racah polynomials and certain ratios of scalar products of Bethe states associated with Bethe equations of either homogeneous or inhomogeneous type. Clearly, eigenbases for other examples of Leonard pairs and related orthogonal polynomials of the discrete Askey-scheme  may 
be studied along the same line as limiting cases. For instance, the Racah type $q=1$ may be studied using the results in \cite{Nico1,Nico2}.\vspace{1mm} 

Some perspectives are now presented. 
Firstly, although not discussed here, let us mention that the modified algebraic Bethe ansatz formalism  applied to the diagonalization of $\pi(\tA^*)$ (or equivalently $\pi(\tA)$)  generates three different Baxter TQ-relations. The first TQ-relation is of homogeneous type, and coincides with the second-order $q$-difference equation for the Askey-Wilson polynomials. As a well-known fact, the zeroes of the Askey-Wilson polynomials are characterized by Bethe equations of homogeneous type. For details, see \cite[Subsection 4.3]{BP19}. The second and third TQ-relations are of inhomogeneous type, and also admit polynomial solutions. In those two cases, the zeroes of the new polynomials satisfy Bethe equations of inhomogeneous type. Due to the fact that each TQ-relation inherits from the theory of Leonard pairs of $q$-Racah type, it suggests that a classification of the corresponding inhomogeneous Bethe equations according to the Askey-scheme - as well as related new polynomials - should be investigated further.\vspace{1mm}

Secondly, we recall that scalar products of Bethe states may be written in terms of
certain compact determinant formulas involving the Bethe roots.
It is certainly worth trying to obtain analogous expressions for the various
types of scalar products entering the $q$-Racah formulae. The best route to that is to
use the modern method introduced in \cite{BS19}. Indeed, this method was successfully applied to compute scalar products similar to those that enter the $q$-Racah formulae \cite{BPS21,BS21}.
\vspace{1mm}

Thirdly, it is known that irreducible finite dimensional representations of the $q$-Onsager algebra \cite{Ter03,Bas04} generated by $\tW_0,\tW_1$ are classified according to the theory of tridiagonal pairs of $q$-Racah type \cite{IT09} that generalizes the theory of Leonard pairs. For a tridiagonal pair of $q$-Racah type, a characterization of the entries of the transition matrices relating the eigenbases of $\tW_0$ to the eigenbases of $\tW_1$ (and more generally the eigenbases of each family of so-called alternating generators \cite{Bas04,BS09,BB17,T21a,T21b}) in terms of special functions is an open problem. Adapting the analysis presented here based on the modified algebraic Bethe ansatz, we expect this problem can be solved in terms of ratios of scalar products of Bethe states associated with the open XXZ spin chain for special boundary conditions.\vspace{1mm} 

Some of these problems will be discussed elsewhere. 

\vspace{0.2cm}

\noindent{\bf Acknowledgments:} We thank S. Belliard, N. Cramp\'e and P. Terwilliger for comments on the manuscript. 
P.B.  is supported by C.N.R.S. R.A.P. acknowledges the support by the German Research Council (DFG) via the Research Unit FOR 2316.
%\vspace{0.2cm}

\begin{appendix}

\section{Askey-Wilson generators and the dynamical operators}\label{apA}
In this section, we recall the precise relationship between the generators of the Askey-Wilson algebra $\tA,\tA^*$ and the so-called 
 dynamical operators $ \{\mathscr{A}^{\epsilon}(u,m), \mathscr{B}^{\epsilon}(u,m),\mathscr{C}^{\epsilon}(u,m),\mathscr{D}^{\epsilon}(u,m)\}$ that arise within the inverse scattering framework. The exchange relations satisfied by the dynamical operators are first recalled. We refer to \cite{BP19} for details. Note that for the analysis in this paper, the left and right actions of the dynamical operators $\mathscr{A}^{\epsilon}(u,m),\mathscr{B}^{\epsilon}(u,m)$ on certain products of $\mathscr{B}^{\epsilon}(u,m)$ and $\mathscr{C}^{\epsilon}(u,m)$ are also needed.

\subsection{The dynamical operators and exchange relations}
Let $m$ be a positive integer and $\epsilon= \pm1$. The dynamical operators $ \{\mathscr{A}^{\epsilon}(u,m), \mathscr{B}^{\epsilon}(u,m),\mathscr{C}^{\epsilon}(u,m),\mathscr{D}^{\epsilon}(u,m)\}$ are subject to the exchange relations:
\ben
\qquad \mathscr{B}^{\epsilon}(u,m+2)\mathscr{B}^{\epsilon}(v,m) &=& \mathscr{B}^{\epsilon}(v,m+2)\mathscr{B}^{\epsilon}(u,m),\label{comBdBd} \\
\qquad \mathscr{A}^{\epsilon}(u,m+2)\mathscr{B}^{\epsilon}(v,m)&=&f(u,v)\mathscr{B}^{\epsilon}(v,m)\mathscr{A}^{\epsilon}(u,m) \label{comAdBd}\\
&& + g(u,v,m)\mathscr{B}^{\epsilon}(u,m)\mathscr{A}^{\epsilon}(v,m) + w(u,v,m)\mathscr{B}^{\epsilon}(u,m)\mathscr{D}^{\epsilon}(v,m),\nonumber  \\
 \mathscr{D}^{\epsilon}(u,m+2)\mathscr{B}^{\epsilon}(v,m)&=& h(u,v)\mathscr{B}^{\epsilon}(v,m) \mathscr{D}^{\epsilon}(u,m) 
,\label{comDdBd}\\
&&+k(u,v,m)\mathscr{B}^{\epsilon}(u,m)\mathscr{D}^{\epsilon}(v,m)+ n(u,v,m)\mathscr{B}^{\epsilon}(u,m)\mathscr{A}^{\epsilon}(v,m),\nonumber\\
\mathscr{C}^{\epsilon}(u,m+2)\mathscr{B}^{\epsilon}(v,m)&=& \mathscr{B}^{\epsilon}(v,m-2)\mathscr{C}^{\epsilon}(u,m)  \label{comcdBd} \\
&& +q(u,v,m) \mathscr{A}^{\epsilon}(v,m)\mathscr{D}^{\epsilon}(u,m)+r(u,v,m)\mathscr{A}^{\epsilon}(u,m)\mathscr{D}^{\epsilon}(v,m)
\nonumber\\
&& +s(u,v,m) \mathscr{A}^{\epsilon}(u,m)\mathscr{A}^{\epsilon}(v,m)+x(u,v,m)\mathscr{A}^{\epsilon}(v,m)\mathscr{A}^{\epsilon}(u,m)
\nonumber\\
&&+y(u,v,m) \mathscr{D}^{\epsilon}(u,m)\mathscr{A}^{\epsilon}(v,m)+z(u,v,m)\mathscr{D}^{\epsilon}(u,m)\mathscr{D}^{\epsilon}(v,m) \nonumber
\een
and
\ben
\qquad \mathscr{C}^{\epsilon}(u,m-2)\mathscr{C}^{\epsilon}(v,m) &=& \mathscr{C}^{\epsilon}(v,m-2)\mathscr{C}^{\epsilon}(u,m),\label{comCdCd} \\
 \qquad  \mathscr{C}^{\epsilon}(v,m+2)\mathscr{A}^{\epsilon}(u,m+2)&=&f(u,v)\mathscr{A}^{\epsilon}(u,m)\mathscr{C}^{\epsilon}(v,m+2)  \label{comAdCd}
\\ && +  g(u,v,m)\mathscr{A}^{\epsilon}(v,m)\mathscr{C}^{\epsilon}(u,m+2) + w(v,u,m)\mathscr{D}^{\epsilon}(v,m)\mathscr{C}^{\epsilon}(u,m+2), \nonumber \\
 \qquad  \mathscr{C}^{\epsilon}(v,m+2)\mathscr{D}^{\epsilon}(u,m+2)&=& h(u,v)\mathscr{D}^{\epsilon}(u,m) \mathscr{C}^{\epsilon}(v,m+2) \label{comDdCd}
\\ &&+k(u,v,m)\mathscr{D}^{\epsilon}(v,m)\mathscr{C}^{\epsilon}(u,m+2)+ n(u,v,m)\mathscr{A}^{\epsilon}(v,m)\mathscr{C}^{\epsilon}(u,m+2)\,\nonumber
\een
where the coefficients are given by
\ben
\nonumber&&f(u,v)= \frac{b(qv/u)b(uv)}{b(v/u)b(quv)}\,,\quad h(u,v)= \frac{b(q^2uv)b(qu/v)}{b(quv)b(u/v)},\\
\nonumber&&g(u,v,m)=\frac{\gamma(u/v,m+1)}{\gamma(1,m+1)}\frac{b(q) b\left(v^2\right)}{b\left(q v^2\right) b\left(\frac{u}{v}\right)},
\quad w(u,v,m)=-\frac{\gamma(uv,m)}{\gamma(1,m+1)}\frac{b(q)}{b(q u v)},\\
\nonumber &&k(u,v,m)=\frac{ \gamma(v/u,m+1)}{\gamma(1,m+1)}\frac{b(q) b\left(q^2 u^2\right)}{b\left(q u^2\right) b\left(\frac{v}{u}\right)}, \quad
n(u,v,m)=\frac{\gamma(1/(uv),m+2)}{\gamma(1,m+1)} \frac{b(q) b\left(v^2\right) b\left(q^2 u^2\right)}{b\left(q u^2\right) b\left(q v^2\right)
   b(q u v)}\,\,,
	\een
\ben
&&q(u,v,m)=\frac{ \gamma \left(u/v,m\right)b(q) b(u v)}{\gamma (1,m+1) b\left(u/v\right) b(q u v)}\,,\quad
r(u,v,m)=\frac{b(q) b\left(u^2\right) \gamma (1,m) \gamma \left(v/u,m+1\right)}{\gamma (1,m+1)^2 b\left(q
   u^2\right) b\left(v/u\right)}\,,\\
\nonumber&&s(u,v,m)=\frac{b(q)^2 b\left(u^2\right) \gamma \left(v^{-2},m+1\right) \gamma \left(v/u,m+1\right)}{\gamma
   (1,m+1)^2 b\left(q u^2\right) b\left(q v^2\right) b\left(\frac{v}{u}\right)}\,,
\quad
x(u,v,m)=\frac{b(q) b\left(u^2\right) b\left(q u/v\right) \gamma \left(1/(uv),m+1\right)}{\gamma (1,m+1)
   b\left(q u^2\right) b\left(u/v\right) b(q u v)}\,,\\
\nonumber&&y(u,v,m)=-\frac{b(q)^2 \gamma \left(v^{-2},m+1\right) \gamma (u v,m)}{\gamma (1,m+1)^2 b\left(q v^2\right) b(q u v)}\,,
\quad
z(u,v,m)=-\frac{b(q) \gamma (1,m) \gamma (u v,m)}{\gamma (1,m+1)^2 b(q u v)}\,.\nonumber
\een
where
\ben
 \gamma^\epsilon (u,m)= \alpha ^{\frac{1-\epsilon }{2}} \beta ^{\frac{\epsilon +1}{2}}
   q^{-m} u -\alpha ^{\frac{\epsilon +1}{2}} \beta ^{\frac{1-\epsilon }{2}}
   q^m u^{-1} \ .\label{gam}
\een

\subsection{Products of dynamical operators and exchange relations}
Using the exchange relations above, the action of the entries on products of $\mathscr{B}^{\epsilon}(u,m),\mathscr{C}^{\epsilon}(u,m)$ can be derived. Recall the strings of length $M$ of operators $ \mathscr{B}^{\epsilon}(u_i,m)$ defined in (\ref{SB}). Define
\ben
B^{\epsilon}(\{u,\bar u_i\},m,M)&=&\mathscr{B}^{\epsilon}(u_1,m+2(M-1))\cdots \mathscr{B}^{\epsilon}(u,m+2(M-i)) \dots \mathscr{B}^{\epsilon}(u_M,m)\,.\label{SB2}
\een
Using the relations (\ref{comBdBd})-(\ref{comDdBd}),
one shows that the action of the dynamical operators
$\{\mathscr{A}^{\epsilon}(u,m),\mathscr{D}^{\epsilon}(u,m)\}$ on the string (\ref{SB}) is given by:
\ben\label{AonSB}
\qquad \mathscr{A}^{\epsilon}(u,m+2M)B^{\epsilon}(\bar u,m,M)&=&
\prod_{i=1}^Mf(u,u_i) B^{\epsilon}(\bar u,m,M)\mathscr{A}^{\epsilon}(u,m)\\
&+&
\sum_{i=1}^M g(u,u_i,m+2(M-1))\! \! \!\prod_{j=1,j\neq i}^M  \! f(u_i,u_j)   B^{\epsilon}(\{u,\bar u_i\},m,M)\mathscr{A}^{\epsilon}(u_i,m)
\nonumber
\\
&+&
\sum_{i=1}^M w(u,u_i,m+2(M-1)) \! \! \!\prod_{j=1,j\neq i}^M  \! h(u_i,u_j)   B^{\epsilon}(\{u,\bar u_i\},m,M)\mathscr{D}^{\epsilon}(u_i,m)
\nonumber
\een
and
 \ben\label{DonSB}
\qquad  \mathscr{D}^{\epsilon}(u,m+2M)B^{\epsilon}(\bar u,m,M)&=&
\prod_{i=1}^M h(u,u_i) B^{\epsilon}(\bar u,m,M)\mathscr{D}^{\epsilon}(u,m)\\\nonumber
&+&\sum_{i=1}^M k(u,u_i,m+2(M-1))\! \! \!\prod_{j=1,j\neq i}^M  \! h(u_i,u_j)   B^{\epsilon}(\{u,\bar u_i\},m,M)\mathscr{D}^{\epsilon}(u_i,m)\\
&+&\sum_{i=1}^M n(u,u_i,m+2(M-1)) \! \! \!\prod_{j=1,j\neq i}^M  \! f(u_i,u_j)   B^{\epsilon}(\{u,\bar u_i\},m,M)\mathscr{A}^{\epsilon}(u_i,m)\,.\nonumber
\een

In addition to (\ref{stringC1}), the following strings of dynamical operators $\mathscr{C}^{\epsilon}(v_i,m)$ with length $N$ is defined:
\ben
C^{\epsilon}(\{v,\bar v_i\},m,N)&=&\mathscr{C}^{\epsilon}(v_1,m+2)\cdots \mathscr{C}^{\epsilon}(v,m+2i) \cdots \mathscr{C}^{\epsilon}(v_N,m+2N)\,.\label{stringC2}
\een
Using the  relations (\ref{comCdCd})-(\ref{comDdCd}),
we can similarly  show that the left action of the `diagonal' dynamical operators
$\{\mathscr{A}^{\epsilon}(v,m),\mathscr{D}^{\epsilon}(v,m)\}$ on the string (\ref{stringC1}) is given by,
\ben\label{AonSC}
&&C^{\epsilon}(\bar v,m,N)\mathscr{A}^{\epsilon}(v,m+2N)=
f(v,\bar v)\mathscr{A}^{\epsilon}(v,m)C^{\epsilon}(\bar v,m,N)\\
&&\qquad\qquad+
\sum_{i=1}^N g(v,v_i,m+2(N-1))f(v_i,\bar v_i)\mathscr{A}^{\epsilon}(v_i,m)C^{\epsilon}(\{v,\bar v_i\},m,N)
\nonumber
\\
&&\qquad\qquad+
\sum_{i=1}^N w(v,v_i,m+2(N-1))h(v_i,\bar v_i)\mathscr{D}^\epsilon(v_i,m)C^{\epsilon}(\{v,\bar v_i\},m,N)
\nonumber
\een
and
\ben\label{DonSC}
&&C^{\epsilon}(\bar v,m,N)\mathscr{D}^\epsilon(v,m+2N)=
h(v,\bar v)\mathscr{D}^\epsilon(v,m)C^{\epsilon}(\bar v,m,N)\\\nonumber
&&\qquad\qquad+\sum_{i=1}^N k(v,v_i,m+2(N-1))h(v_i,\bar v_i)\mathscr{D}^\epsilon(v_i,m)C^{\epsilon}(\{v,\bar v_i\},m,N)\\
&&\qquad\qquad+\sum_{i=1}^N n(v,v_i,m+2(N-1))f(v_i,\bar v_i)\mathscr{A}^\epsilon(v_i,m)C^{\epsilon}(\{v,\bar v_i\},m,N).\nonumber
\een

\subsection{Expressions of $\tA,\tA^*$ in terms of the dynamical operators}
 The expression of the Askey-Wilson algebra generators in terms of the dynamical operators previously introduced is now recalled.  As shown in \cite{BP19},  according to the choice $\epsilon=\pm 1$ the element $\tA$ can be expressed in two different ways. For instance, in terms of the dynamical operators for $\epsilon=-1$ one has:
\beqa
\label{defAepm}\textsf{A}&=&{\mathbb{A}}^-(u,m)  +\frac{\left(q\,u\,\bar\eta(u)+q^{-1}u^{-1}\bar\eta(u^{-1})\right)}{(u^2-u^{-2})(q^2u^2-q^{-2}u^{-2})}\quad \mbox{with}\quad 
\bar\eta(u)=(q+q^{-1})\rho^{-1}\left(\eta u+\eta^*u^{-1}\right)\, ,
\eeqa
where, for further convenience, we have introduced:
\beqa
&&\label{Aepm}{\mathbb{A}}^-(u,m)= \frac{u^{-1}}{(u^2-u^{-2})}\left(
\frac{1}{(qu^2-q^{-1}u^{-2})}\mathscr{A}^{-}(u,m) +
\frac{1}{(q^2u^2-q^{-2}u^{-2})}\mathscr{D}^{-}(u,m)\right)\,.
\eeqa

Alternatively, for $\epsilon=+1$ a slightly more complicated expression in terms of the dynamical operators  is obtained. Namely:
\beqa\label{defAepp}
\textsf{A}&=&\tilde{\mathbb{A}}^+(u,m)  +\frac{\left(q\,u\,\bar\eta(u)+q^{-1}u^{-1}\bar\eta(u^{-1})\right)}{(u^2-u^{-2})(q^2u^2-q^{-2}u^{-2})}\ ,\
\eeqa
where (recall the definition (\ref{gam}))
\beqa
\label{Asepm}\tilde{\mathbb{A}}^+(u,m)&=&\frac{u}{(u^2-u^{-2})}     \left( \frac{\gamma^+\left(q^{-1} u^{-2},m  \right)}
{(qu^2-q^{-1}u^{-2}) \gamma^+(1,m+1)} 
\mathscr{A}^{+}(u,m) +
\frac{\gamma^+ \left(q u^2,m\right)}{ (q^2u^2-q^{-2}u^{-2})  \gamma^+(1,m+1)}
\mathscr{D}^{+}(u,m)\right.
\nonumber\\
&&\qquad\qquad\qquad   \left.+
\frac{\alpha  q^{m+1}}
{ \gamma^+(1,m)} \mathscr{B}^{+}(u,m)- \frac{\beta  q^{1-m}} {\gamma^+ (1,m)} \mathscr{C}^{+}(u,m)\right)\ . \nonumber\,
\eeqa

In the main text, we also need the expression of $\tA^*$ for the cases $\epsilon=\pm 1$. For $\epsilon=+ 1$, it reads:
\beqa
\label{defAsepp}\textsf{A}^*&=& {\mathbb{A}}^+(u,m)  +\frac{\left(q\,u\,\bar\eta(u^{-1})+q^{-1}u^{-1}\bar\eta(u)\right)}{(u^2-u^{-2})(q^2u^2-q^{-2}u^{-2})} \, 
\eeqa
where we have introduced:
\beqa
&&\label{Asepp}
{\mathbb{A}}^+(u,m)=  \frac{u}{(u^2-u^{-2})}\left(
\frac{1}{(qu^2-q^{-1}u^{-2})}\mathscr{A}^{+}(u,m) +
\frac{1}{(q^2u^2-q^{-2}u^{-2})}\mathscr{D}^{+}(u,m)\right)\,.
\eeqa

For $\epsilon=- 1$, one has:
\beqa\label{defAsepm}
\tA^*=\tilde{\mathbb{A}}^-(u,m) +\frac{\left(q\,u\,\bar\eta(u^{-1})+q^{-1}u^{-1}\bar\eta(u)\right)}{(u^2-u^{-2})(q^2u^2-q^{-2}u^{-2})} 
\eeqa
with
\beqa
\label{Asepm}\tilde{\mathbb{A}}^-(u,m)&=&\frac{u^{-1}}{(u^2-u^{-2})}     \left( \frac{\gamma^-\left(q^{-1} u^{-2},m  \right)}
{(qu^2-q^{-1}u^{-2}) \gamma^-(1,m+1)} 
\mathscr{A}^{-}(u,m) +
\frac{\gamma^- \left(q u^2,m\right)}{ (q^2u^2-q^{-2}u^{-2})  \gamma^-(1,m+1)}
\mathscr{D}^{-}(u,m)\right.
\nonumber\\
&&\qquad\qquad\qquad   \left.+
\frac{\alpha  q^{-m-1}}
{ \gamma^-(1,m)} \mathscr{B}^{-}(u,m)- \frac{\beta  q^{m-1}} {\gamma^- (1,m)} \mathscr{C}^{-}(u,m)\right)\ \nonumber\ .
\eeqa

The expressions of $\tA,\tA^*$ in terms of the dynamical operators together with the exchange relations above are used to compute the action of $\tA,\tA^*$ on the so-called Bethe states.
\vspace{3mm}

\end{appendix}

\end{document}